\documentclass[11pt]{article}
\usepackage[margin=1in]{geometry}

\usepackage[utf8]{inputenc}
\usepackage[breakable, theorems, skins]{tcolorbox}
\tcbset{enhanced}
\allowdisplaybreaks

\usepackage{amsfonts}

\usepackage{parskip}
\usepackage{multirow}

\usepackage[suppress]{color-edits}
\addauthor{gl}{purple}
\addauthor{nsm}{red}
\addauthor{yy}{blue}

\usepackage{enumitem}

\usepackage{amsmath,amssymb,amsthm}
\usepackage{bm}
\usepackage{amsmath}[fleqn]
\usepackage{amsthm,amssymb}
\usepackage{xcolor}
\usepackage{mathtools}
\usepackage{enumitem}
\usepackage{bbm}
\usepackage{semantic}
\usepackage{listings}
\usepackage{courier}
\usepackage{tikz}
\usetikzlibrary{calc}
\usepackage{color}
\usepackage{fancyhdr}
\usepackage{lastpage}
\usepackage{thm-restate}

\definecolor{codegreen}{rgb}{0,0.6,0}
\definecolor{codegray}{rgb}{0.5,0.5,0.5}
\definecolor{codepurple}{rgb}{0.58,0,0.82}
\definecolor{backcolour}{rgb}{0.95,0.95,0.92}
 
\lstdefinestyle{mystyle}{
    backgroundcolor=\color{backcolour},   
    commentstyle=\color{codegreen},
    keywordstyle=\color{magenta},
    numberstyle=\tiny\color{codegray},
    stringstyle=\color{codepurple},
    basicstyle=\footnotesize,
    breakatwhitespace=false,         
    breaklines=true,                 
    captionpos=b,                    
    keepspaces=true,                 
    numbers=left,                    
    numbersep=5pt,                  
    showspaces=false,                
    showstringspaces=false,
    showtabs=false,                  
    tabsize=2,
    basicstyle=\footnotesize\ttfamily
}

\definecolor{asparagus}{rgb}{0.53, 0.66, 0.42}
\definecolor{sacramentostategreen}{rgb}{0.0, 0.3, 0.15}
\definecolor{teal}{rgb}{0.0, 0.5, 0.5}
\definecolor{forestgreen}{rgb}{0.13, 0.6, 0.13}

\usepackage{hyperref}
\hypersetup{
    colorlinks=true,
    linkcolor=sacramentostategreen,
    filecolor=sacramentostategreen,
    citecolor=teal,
    urlcolor=teal,
}

\let\Pr\relax
\DeclareMathOperator*{\Pr}{\mathrm{Pr}}

\newcommand{\eps}{\epsilon}
\renewcommand{\S}{\mathbb{S}}
\newcommand{\paren}[1]{\left(#1\right)}

\newcommand{\sqb}[1]{\left[#1\right]}

\newcommand{\ceil}[1]{\left\lceil#1\right\rceil}

\newcommand{\myfunc}[3]{#1\colon#2\to#3}

\newtheorem{theorem}{Theorem}
\newtheorem{corollary}[theorem]{Corollary}

\newtheorem{lemma}[theorem]{Lemma}
\newtheorem{definition}{Definition}

\newtheorem{fact}[theorem]{Fact}

\newtheorem{problem}{Problem}

\numberwithin{exercise}{section}
\newcommand{\exv}[1]{\E\sqb{#1}}

\newcommand{\expv}[1]{\mathsf{exp}\paren{#1}}
\newcommand{\prv}[1]{\Pr\sqb{#1}}
\newcommand{\prvv}[2]{\underset{#1}{\Pr}\sqb{#2}}
\newcommand{\logv}[1]{\log\paren{#1}}

\renewcommand{\vec}[1]{\overrightarrow{#1}}

\makeatletter
\newcommand{\vo}{\vec{o}\@ifnextchar{^}{\,}{}}
\makeatother

\usepackage{tcolorbox}
\usepackage{algorithm}
\usepackage{algorithmic} % # Chloe:changed to algpseudocode
\newlength\tindent
\numberwithin{equation}{section}

\def\to{\rightarrow}
\def\eps{\varepsilon}
\def\epsilon{\varepsilon}

\def\eps{\epsilon}

\def\phi{\varphi}

\newcommand{\R}{{\mathbb R}}
\newcommand{\E}{{\mathbb E}}

\newcommand{\N}{{\mathbb{N}}}

\newcommand{\indicator}[1]{\mathbbm{1}\inbraces{#1}}

\renewcommand{\Pr}{\mathsf{Pr}}

\usepackage{nicefrac}
\let\nfrac=\nicefrac
\newcommand{\abs}[1]{\ensuremath{\left\lvert #1 \right\rvert}}

\newcommand{\norm}[1]{\ensuremath{\left\lVert #1 \right\rVert}}

\newcommand{\ip}[1]{\left\langle #1 \right\rangle}

\newfont{\inhead}{eufm10 scaled\magstep1}

\newcommand{\suchthat}{{\;\; : \;\;}}

\newcommand{\argmax}{\mathrm{argmax}}
\newcommand{\argmin}{\mathrm{argmin}}

\newcommand{\inparen}[1]{\left(#1\right)}             %\inparen{x+y}  is (x+y)
\newcommand{\inbraces}[1]{\left\{#1\right\}}           %\inbrace{x+y}  is {x+y}
\newcommand{\insquare}[1]{\left[#1\right]}             %\insquare{x+y}  is [x+y]
\def\to{\rightarrow}
\def\eps{\varepsilon}
\def\epsilon{\varepsilon}

\def\eps{\epsilon}

\def\phi{\varphi}

\newcommand{\cA}{\mathcal{A}}
\newcommand{\cB}{\mathcal{B}}

\newcommand{\cD}{\mathcal{D}}

\newcommand{\cM}{\mathcal{M}}

\newcommand{\cX}{\mathcal{X}}

\newcommand{\mI}{\mathbf{I}}

\newcommand{\ve}{\bm{e}}

\newcommand{\vg}{\bm{g}}

\newcommand{\vr}{\bm{r}}
\newcommand{\vs}{\bm{s}}

\newcommand{\vx}{\bm{x}}
\newcommand{\vy}{\bm{y}}
\newcommand{\vz}{\bm{z}}

\usepackage{multirow}
\usepackage{soul}

\usepackage{thmtools}
\usepackage{thm-restate}

\definecolor{ForestGreen}{RGB}{34, 139, 34}

\usepackage[
backref=true,
natbib=true,
style=trad-alpha,
maxalphanames=8,
sorting=nyt
]{biblatex}
\newcommand{\xstar}{\vr}
\newcommand{\zstar}{\vx^{\star}} % \vx^{\star}
\newcommand{\xt}[1]{\vx_t^{(#1)}}
\newcommand{\actionset}{\cX}  % the original action set the algorithm is given
\newcommand{\xplus}{\vx^{+}}

\newif\ifnotes
\makeatletter
\newcommand{\nnote}[1]{\@bsphack\ifnotes{$\ll$\textsf{\color{blue} Naren: { #1}}$\gg$}\fi\@esphack}
\newcommand{\gl}[1]{\@bsphack\ifnotes{\color{red}{[GL: #1]}}\fi\@esphack}
\newcommand{\yy}[1]{\@bsphack\ifnotes{\color{violet}{[CY: #1]}}\fi\@esphack}
\newcommand{\as}[1]{\@bsphack\ifnotes{\color{blue}{[AS: #1]}}\fi\@esphack}
\newcommand{\ab}[1]{\@bsphack\ifnotes{\color{mahogany}{[AB: #1]}}\fi\@esphack}
\newcommand{\mg}[1]{\@bsphack\ifnotes{\color{purple}{[MG: #1]}}\fi\@esphack}
\makeatother
\notesfalse

\newcommand\numberthis{\addtocounter{equation}{1}\tag{\theequation}}

\begin{document}

% \author{
% Avrim Blum
% \thanks{Toyota Technological Institute at Chicago. Email: \url{avrim@ttic.edu}. Supported by NSF Awards CCF-2212968
% and ECCS-2216899 and by the Defense Advanced Research Projects Agency under cooperative
% agreement HR00112020003.}
% \and
% Meghal Gupta
% \thanks{UC Berkeley. Email: \url{meghal@berkeley.edu}. Supported by NSF Graduate Research Fellowship.}
% \and
% Gene Li
% \thanks{Toyota Technological Institute at Chicago. Email: \url{gene@ttic.edu}. Supported by the Institute for Data, Econometrics, Algorithms, and Learning (IDEAL).}
% \and
% Naren Sarayu Manoj
% \thanks{Toyota Technological Institute at Chicago. Email: \url{nsm@ttic.edu}. Supported by NSF Graduate Research Fellowship and NSF Award ECCS-2216899.}
% \and 
% Aadirupa Saha\thanks{Toyota Technological Institute at Chicago (currently at Apple Research). Email: \url{aadirupa@ttic.edu}.}
% \and
% Yuanyuan Yang
% \thanks{University of Washington. Email: \url{yyangh@cs.washington.edu}. Supported by NSF grant No. CCF-2045402 and NSF grant No. CCF-2019844.}
% }

\author{
Avrim Blum
\thanks{Toyota Technological Institute at Chicago. Email: \url{avrim@ttic.edu}.}
\and
Meghal Gupta
\thanks{UC Berkeley. Email: \url{meghal@berkeley.edu}.}
\and
Gene Li
\thanks{Toyota Technological Institute at Chicago. Email: \url{gene@ttic.edu}.}
\and
Naren Sarayu Manoj
\thanks{Toyota Technological Institute at Chicago. Email: \url{nsm@ttic.edu}.}
\and 
Aadirupa Saha\thanks{Toyota Technological Institute at Chicago (currently at Apple Research). Email: \url{aadirupa@ttic.edu}.}
\and
Yuanyuan Yang
\thanks{University of Washington. Email: \url{yyangh@cs.washington.edu}.}
}

\title{Dueling Optimization with a Monotone Adversary}
\date{\today}

\notestrue % COMMENT OUT BEFORE SUBMISSION

\maketitle

% in the abstract, don't use macros. we will have to copy paste this into various places and it's way easier if there are no macros
\begin{abstract}
We introduce and study the problem of \textit{dueling optimization with a monotone adversary}, which is a generalization of (noiseless) dueling convex optimization. The goal is to design an online algorithm to find a minimizer $\bm{x}^{\star}$ for a function $f\colon \mathcal{X} \to \mathbb{R}$, where $\mathcal{X} \subseteq \mathbb{R}^d$. In each round, the algorithm submits a pair of guesses, i.e., $\bm{x}^{(1)}$ and $\bm{x}^{(2)}$, and the adversary responds with \textit{any} point in the space that is at least as good as both guesses. The cost of each query is the suboptimality of the worse of the two guesses; i.e., ${\max} \left( f(\bm{x}^{(1)}), f(\bm{x}^{(2)}) \right) - f(\bm{x}^{\star})$. The goal is to minimize the number of iterations required to find an $\eps$-optimal point and to minimize the total cost (regret) of the guesses over many rounds. Our main result is an efficient randomized algorithm for several natural choices of the function $f$ and set $\mathcal{X}$ that incurs cost $O(d)$ and iteration complexity $O(d\log(1/\varepsilon)^2)$. Moreover, our dependence on $d$ is asymptotically optimal, as we show examples in which any randomized algorithm for this problem must incur $\Omega(d)$ cost and iteration complexity.
\end{abstract}
\newpage
\section{Introduction}
A growing body of literature studies learning with preference-based feedback \citep{bv06, shivaswamy2011online}, with tremendous empirical success in recommendation systems, search engine optimization, information retrieval, and robotics. More recently, preference-based feedback has received a lot of attention as a mechanism to train large language models \citep{ouyang2022training}. Moreover, in recommender systems \citep{bobadilla2013recommender}, a natural approach is to learn from users' preferences relations on a set of recommended items and update the system's belief for better future recommendations~\cite{j16} (e.g., given these items, which one do you prefer the most?).

Such preference-based feedback is not readily addressed by classical formulations for online decision making, such as bandits and reinforcement learning. In particular, algorithms for these problems rely on ordinal feedback per item (e.g., on a scale of 1 to 10, how much did the user like a particular item?). To address this, a long line of work studies the \emph{dueling bandit framework} for online decision making under pairwise/preference-based feedback. There exist efficient algorithms with provable guarantees for the standard multi-armed bandit setup \citep{yue2012k, ailon2014reducing, komiyama2015regret}, contextual bandits \cite{dudik2015contextual,saha2022efficient}, as well as dueling convex optimization \cite{jamieson2012query,skm21,saha2022dueling}, to name a few. The dueling bandit framework is especially applicable in settings where real-valued feedback is scarce or impossible to obtain, but preference-based feedback is readily available.

However, a key limitation of the dueling bandit framework is that the feedback that the learner receives is essentially ``in-list''. That is, the users are restricted to selecting items exclusively from the list of recommended items. This feedback model fails to capture the real-world scenarios where the users might select an out-of-list item they prefer. To illustrate, music streaming services like Spotify create personalized playlists for users. Concretely, each song can be encoded as a feature vector $\vx \in \R^d$, and the goal is to recommend the songs with the highest utility for a hidden, well-structured utility function of $\vx$. However, the users can also search for and play the songs they have a stronger preference (i.e., higher utility) than all recommendations. 

This out-of-list feedback model falls into a monotone adversarial framework (see the chapter by \citet{feige_2021}). In such models, an adversary is only allowed to make ``helpful'' changes. For example, in a graph clustering problem, the adversary is only allowed to add edges within communities and delete edges that cross communities (see, e.g., the chapter by \citet{moitra_2021}). In our setting, the adversary is only allowed to respond with an item that is at least as good as any recommended item. A clear adaptation of the dueling bandit framework to this new feedback type is not evident. 
%To illustrate, consider the scenario of Netflix movie recommendation. When a user wants to watch a movie on Netflix, the recommendation system suggests a list of choices. Concretely, we can model each movie via a feature vector $\vx \in \R^d$ and the user's utility for watching a given movie as an unknown linear function of $\vx$; the goal of the recommendation system is to recommend the movies with the highest utility by observing which movies the user choose to watch. However, the user may reject everything in the list and instead watch a different movie of their choice. It is not clear how to adapt the dueling bandit framework to handle this type of feedback.

\subsection{Problem statement}

% In this work, we introduce the first theoretical formulation for this setting, named \emph{dueling optimization with a monotone adversary}. \yy{I added this sentence. }Note that the``in-list" feedback remains to be a feasible strategy, hence any solution for this problem handles both ``in-list" and ``out-of-list" feedback effectively.
% As we will see, our formulation supports ``out-of-list'' feedback.

As our main conceptual contribution, we introduce a theoretical formulation for this setting that we call \textit{dueling optimization with a monotone adversary}. As we will see, our formulation supports ``out-of-list'' feedback.

\begin{problem}[Dueling optimization with a monotone adversary]
\label{problem:local_rec_fixed_set} % For any prespecified $\eps \in (0,1)$, the goal is to find an $\eps$-optimal point $\widehat{\vx} \in \mathcal{X}$ such that $f(\widehat{\vx}) - f(\zstar) \le \eps$ using queries $\inbraces{(\xt{1}, \xt{2} )}_{t=1}^T$ of the following form:
Let $\mathcal{X} \subseteq \R^d$ be a decision space, and let $\myfunc{f}{\mathcal{X}}{\R}$ be a cost function with an unknown global minimum $\zstar$. A learner interacts with an adversary over rounds $t = 1, 2, \dots$, where each round is of the following form.
\begin{enumerate}
    \item The learner proposes two points $\xt{1}, \xt{2} \in \mathcal{X}$.
    \item The adversary responds with a point $\xstar_t$ that satisfies
    $f(\xstar_t) \le \min \inbraces{f\inparen{\xt{1}}, f\inparen{\xt{2}}}\label{eq:valid_feedback}$.
\end{enumerate}
The goal is to design algorithms that:
\begin{enumerate}
    \item for some prespecified $\eps > 0$, minimize the number of iterations to find a point $\vx$ for which $f(\vx) - f(\zstar) \le \eps$;
    \item minimize the total \textit{cost} $\sum_{t = 1}^{\infty} \inparen{\max\inbraces{f\inparen{\xt{1}}, f\inparen{\xt{2}}} - f(\zstar)}$.
\end{enumerate}
\end{problem}

Note that in Problem \ref{problem:local_rec_fixed_set}, we are interested in both the iteration complexity and the total cost. The first objective is a standard metric for measuring the performance of an iterative optimization algorithm. The second objective is motivated by online settings in which a practitioner may wish to minimize the total regret (cost) of its recommendations over an indefinitely long interaction with a user. %(as the described interaction model reflects a natural online learning scenario). % Additionally, an algorithm that adequately addresses one of the two objectives may not immediately address the other. For
% Thus, our problem is a bicriteron optimization problem, where optimizing one objective (even approximately) may compromise the other.\footnote{More generally, this problem falls under multi-objective optimization, where the solution concept is the Pareto Frontier rather than a single point(See~\cite{d04} for an overview).}
%It is not immediate that an algorithm that accomplishes one objective gives an algorithm for the other, but 
In fact, the algorithms we propose in this paper simultaneously achieve both small iteration complexity (for any choice of $\eps$) as well as total cost --- see our technical overview in Section \ref{sec:tech-overview} for more details.

Problem \ref{problem:local_rec_fixed_set} is a natural extension of (noiseless) dueling optimization \cite{jamieson2012query,skm21,saha2022dueling} to handle ``out-of-list'' responses, as in the Spotify recommendation example. The vanilla (noiseless) dueling optimization setup corresponds to the requirement that the user's response satisfies $\xstar_t \in \{\xt{1}, \xt{2}\}$. We allow the user to be potentially adversarial by allowing it to respond with any improvement to the learner's suggestions (in the sequel, we exclusively refer to the user as the adversary).

Even though the monotone adversary is only improving upon the learner's suggestions, existing algorithms for dueling optimization cannot be freely extended to handle the monotone feedback. At a high level, the difficulty arises from the fact that existing algorithms carefully select the queries $\xt{1}, \xt{2}$ so that learning whether $f(\xt{1}) > f(\xt{2})$ reveals information about the underlying $f$. However, a monotone adversary can return a point $\xstar_t$ that reveals no information about the relationship between $\xt{1}$ and $\xt{2}$.

To illustrate this point, consider a natural coordinate-wise binary search algorithm for the dueling optimization problem when $f(\vx) = \norm{\vx-\zstar}_2^2$ for some $\zstar \in \cB_2^d \coloneqq \inbraces{\vx \suchthat \norm{\vx}_2 \le 1}$. For coordinates $i = 1, \cdots, d$, query points of the form $\xt{1} = c_1 \cdot \ve_i$, $\xt{2} = c_2 \cdot \ve_i$ and progressively refine the values $c_1, c_2 \in \R$ to search for the value of $\zstar[i]$ (i.e., the $i$-th entry of $\zstar$). It is easy to show that this approach has a query complexity of $O\inparen{d \logv{\nfrac{1}{\eps}}}$ in the vanilla dueling optimization setting. However, a monotone adversary can return orthogonal responses of the form $\xstar_t = C\ve_j$ (where $j \ne i$ and $C$ is a constant) that do not allow the learner to search along the intended coordinate $i$. Furthermore, \citet{jamieson2012query} and \citet{skm21} give more sophisticated algorithms for the dueling optimization problem that inherently depend upon the ``in-list'' feedback, which clearly cannot apply to our setting. We therefore need novel insights to solve Problem \ref{problem:local_rec_fixed_set}.

\subsection{Our results}
We study Problem \ref{problem:local_rec_fixed_set} for various natural classes of functions $f$ and provide tight upper and lower bounds on the number of queries required to find an $\eps$-optimal point.

\paragraph{Upper bound for linear functions.} First, we study dueling optimization with a monotone adversary when the function $f$ is linear. This is a natural class to consider. In particular, an algorithm that solves Problem \ref{problem:local_rec_fixed_set} can be adapted to achieve constant regret for (noiseless) linear contextual bandits \cite{chu2011contextual}, where the reward function is $r(\vx) \coloneqq \ip{\vx, \vx^\star}$. Note that the key difference in the setup is that the learner does not get to observe the actual linear costs but instead only an improvement to the actions (points) that the learner selects. % \yy{Should we cite some paper or add a reduction in the appendix?}

\begin{theorem}
\label{cor:bin_search_ip}
Let $\mathcal{X} = \mathbb{S}_2^d$, let $\zstar$ be such that $\norm{\zstar}_2 = 1$, and let $\myfunc{f}{\mathcal{X}}{\R}$ be $f(\vx) = -\ip{\vx,\zstar}$. Fix any $\eps > 0$. There exists an algorithm that, in the setting of Problem \ref{problem:local_rec_fixed_set}, with probability at least $1-\expv{-O(d)}$:
\begin{itemize}
    \item outputs a point $\vx$ satisfying $\ip{\zstar-\vx,\zstar} \le \eps$ within $O(d\logv{\nfrac{1}{\eps}}^2)$ iterations;
    \item incurs total cost $O(d)$.
\end{itemize}
Each pair of guesses at time $t$ can be computed in $O(d)$ time.
\end{theorem}

We prove Theorem \ref{cor:bin_search_ip} in Section \ref{subsec:ip_proof}, and the cost is near-optimal with respect to $d$. 

% \yy{From my understanding of the cited paper, we're studying the same problem as theirs, and their construction of the lower bound only uses vectors in $B_d$, thus should work for our setting. In this sense, either our theorem 1 is wrong, or their theorem 6.4 is wrong? }
\citet{gollapudi2021contextual} study a closely related setup that they call \emph{local contextual recommendation}. Their result (see their Theorem 6.4) can be interpreted as showing that if the action set $\actionset$ is a discrete set (namely a packing over the unit sphere), there exists a $2^{\Omega(d)}$ lower bound on the iteration complexity to find a point with constant suboptimality. In contrast, our Theorem \ref{cor:bin_search_ip} shows a much smaller upper bound when the domain is the entire unit sphere. 

\paragraph{Upper bound for smooth and P\L{} functions.} Next, we study whether we can show guarantees for a large class of functions. We show a positive result for functions that are both $\beta$-smooth and $\alpha$-Polyak-\L{}ojasiewicz (abbreviated as P\L{}). These assumptions are standard in optimization.

\begin{definition}[$\beta$-smooth function {\cite[{Lemma 3.4}]{bubeck2015convex}}]
\label{defn:smoothness}
We say $f$ is $\beta$-smooth if it satisfies (\ref{eq:smoothness}).
\begin{align}
    \text{For all } \vx, \vy \in \R^d &:\quad \abs{f(\vx)-f(\vy) - \ip{\nabla f(\vy), \vx-\vy}} \le \frac{\beta}{2} \cdot \norm{\vx-\vy}_2^2\label{eq:smoothness}
\end{align}
\end{definition}
\begin{definition}[$\alpha$-P\L{} function]
\label{defn:pl}
We say $f$ is $\alpha$-P\L{} if it satisfies (\ref{eq:pl}).
\begin{align}
    \text{ For all } \vx \in \R^d \text{ and minimizers } \zstar &:\quad f(\vx) -f(\zstar) \le \frac{1}{2\alpha}\norm{\nabla f(\vx)}_2^2 \label{eq:pl}
\end{align}
\end{definition}
Our main result for this setting is Theorem \ref{cor:bin_search_pl}.

\begin{theorem}
\label{cor:bin_search_pl}
Let $\mathcal{X} = \R^d$, and suppose $f$ is $\beta$-smooth (Definition \ref{defn:smoothness}) and $\alpha$-P\L{} (Definition \ref{defn:pl}). Fix any $\eps > 0$, as well as a known point $\vx_1$ and a value $B$ satisfying $B \ge f(\vx_1) - f(\zstar)$. There exists an algorithm that, in the setting of Problem \ref{problem:local_rec_fixed_set}, with probability at least $1-\expv{-O(d)}$:
\begin{itemize}
    \item outputs a point $\vx$ satisfying $f(\vx) - f(\zstar) \le \eps$ within $O\inparen{\nfrac{\beta}{\alpha} \cdot d \cdot \logv{\nfrac{B}{\eps}}^2}$ iterations;
    \item incurs total cost $O\inparen{\nfrac{\beta}{\alpha} \cdot B \cdot d}$.
\end{itemize}
Each pair of guesses at time $t$ can be computed in $O(d)$ time.
\end{theorem}

We prove Theorem \ref{cor:bin_search_pl} in Section \ref{subsec:pl_proof}. % It is an interesting future direction to determine whether similar results to our Theorem \ref{cor:bin_search_pl} can be obtained without assuming that $f$ is $\beta$-smooth or $\alpha$-P\L{}.

As an application, we show a positive result when the loss function is the Euclidean distance, and the decision space $\mathcal{X} = \cB_2^d$ is a unit ball: 

\begin{theorem}
\label{cor:bin_search_dist}
Let $\mathcal{X} = \cB_2^d$, let $\zstar$ be such that $\norm{\zstar}_2 \le 1$, and let $\myfunc{f}{\mathcal{X}}{\R}$ be $f(\vx) = \norm{\vx-\zstar}_2$. Fix any $\eps > 0$. There exists an algorithm that, in the setting of Problem \ref{problem:local_rec_fixed_set}, with probability at least $1 - \expv{-O(d)}$:
\begin{itemize}
    \item outputs a point $\vx$ satisfying $\norm{\vx-\zstar}_2 \le \eps$ within $O\inparen{d \cdot \logv{\nfrac{B}{\eps}}^2}$ iterations;
    \item incurs total cost $O\inparen{d}$.
\end{itemize}
Each pair of guesses at time $t$ can be computed in $O(d)$ time.
\end{theorem}

We prove Theorem \ref{cor:bin_search_dist} in Section \ref{subsec:dist_proof}.

Note that unlike in Theorem \ref{cor:bin_search_pl}, Theorem \ref{cor:bin_search_dist} applies to the setting where the algorithm must guess points belonging to a given constraint set $\actionset$. Hence, in the proof of Theorem \ref{cor:bin_search_dist}, we have to be careful to ensure that the convergence argument still holds when we apply the algorithm for Theorem \ref{cor:bin_search_pl} along with a projection step. It is not clear that this argument holds by default for all $f$ satisfying the conditions requested by Theorem \ref{cor:bin_search_pl}. Furthermore, as will become evident, we really only require that $\actionset$ be any convex body (though we state the result with $\actionset = \cB_2^d$ to emphasize the consistency with our following lower bounds).

\paragraph{Lower bounds.} We also prove that the dependence on $d$ in our results is tight. In particular, when $f$ is either a linear function or the distance to the target (as in Theorem~\ref{cor:bin_search_dist}), then $\Omega(d)$ queries are necessary to identify $\zstar$. This will translate to a $\Omega(d)$ cost over an infinite number of rounds. In fact, our lower bound is valid when the adversary must return one of the two queried points, as in vanilla dueling optimization framework. 
% Though we do not explicitly show a lower bound for the functions described in Theorem~\ref{cor:bin_search_ip} and Theorem~\ref{cor:bin_search_pl}, similar lower bounds hold in these settings. \glcomment{do we have a lower bound for linear? Maybe we actually get this for the same construction! we should write this down.}

Our lower bound also covers a more general setting than that stated in Problem \ref{problem:local_rec_fixed_set}. Thus far, we have only discussed the setting where the algorithm can query only two points and is told the better of the two. In many practical instances, the algorithm can query $m$ points and learn the point with the best objective value (we call this $m$-ary dueling optimization). Then, one may ask why we study only the $m=2$ in this paper. In our construction, we prove that unless $m$ is polynomial in $d$, we cannot decrease the total cost substantially below $\Omega(d)$. Thus, the most interesting case for constant $m$ is when $m = 2$, and in this case, our result is optimal in $d$.

See Theorem \ref{thm:large_m_lower} for a formal statement of our lower bound.

\begin{theorem}[Lower bound, $\ell_2$ distance]
\label{thm:large_m_lower}
Let $\cX = \cB_2^d$. For any randomized algorithm for $m$-ary dueling optimization, there exists a choice of minimizer $\zstar \in \cB_2^d$ and function $f(\vx) \coloneqq \norm{\vx-\zstar}_2$ such that the algorithm must:
\begin{itemize}
    \item perform $\Omega\inparen{\nfrac{d}{\log m}}$ iterations in expectation to find a point $\vx$ for which $f(\vx) - f(\zstar) \le \eps$.
    \item incur cost  $\Omega\inparen{\nfrac{d}{\log m}}$ in expectation.
\end{itemize}
Here, $\eps > 0$ is an absolute numerical constant.
\end{theorem}

% \begin{theorem}[Lower bound for $m$-ary Search]
% \label{thm:large_m_lower}
% Let $f(\vx) = \nfrac{1}{2} \cdot \norm{\vx-\zstar}_2^2$, which is $1$-smooth (Definition \ref{defn:smoothness}) and $1$-P\L{} (Definition \ref{defn:pl}). For any randomized algorithm for $m$-ary search in $d$ dimensions, there exists a choice of the minimizer $\zstar \in \cB_2^d$ such that the algorithm must:
% \begin{itemize}
%     \item perform $\Omega\inparen{\nfrac{d}{\log m}}$ iterations in expectation to find a point $\vx$ for which $\nfrac{1}{2} \cdot \norm{\vx - \zstar}_2^2 \le 5 \cdot 10^{-3}$;
%     \item incur cost  $\Omega\inparen{\nfrac{d}{\log m}}$ in expectation.
% \end{itemize}
% \end{theorem}

We prove Theorem \ref{thm:large_m_lower} in Section \ref{sec:lower-bound}. Using the same construction, we can also demonstrate that Theorem \ref{cor:bin_search_ip} is tight when $\mathcal{X}$ is the unit sphere.

\begin{corollary}[Lower bound, linear $f$]\label{corr:inner-prod-lower-bound}
Let $\mathcal{X} = \mathbb{S}_2^d$. For any randomized algorithm for $m$-ary dueling optimization there exists a choice of minimizer $\zstar \in \mathbb{S}_2^d$ and function $f(\vx) \coloneqq -\ip{\vx,\zstar}$ such that the same conclusions as in Theorem \ref{thm:large_m_lower} hold.
\end{corollary}

\subsection{Technical overview}\label{sec:tech-overview}

At a high level, our algorithms maintain a guess $\vx_t$ for the optimal solution $\zstar$. They will update this guess over many interactions with the adversary.

\paragraph{A general recipe.} We first describe the primitives that our methods depend on. Our first technical innovation is the notion of \textit{progress distributions}. Loosely speaking, these are distributions from which a learner is likely to sample a new guess $\vx_{t+1}$ that decreases its suboptimality. See Definition \ref{defn:prog}.

\begin{definition}[Progress Distribution]
\label{defn:prog}
Let $\myfunc{f}{\actionset}{\R}$ for $\actionset \subseteq \R^d$. For $\vx\in \actionset$ and $1 \le p < 2$, we say a distribution $\cD(\vx)$ over vectors in $\R^d$ is a \textit{$(p, \gamma,\rho)$-progress distribution for $\vx$} if we have the below.
\begin{align*}
    \prvv{\xplus \sim \cD(\vx)}{\frac{f(\vx)-f(\xplus)}{\inparen{f(\vx)-f(\zstar)}^p} \ge \frac{\rho}{d}} \ge \gamma.
\end{align*}
\end{definition}

So, if for every $\vx_t$ the learner had sample access to some progress distribution $\cD(\vx_t)$, the learner can significantly improve its solution (e.g. when $p = 1$, roughly $\sim d/\rho$ steps are sufficient for the learner to decrease its suboptimality by a constant factor). It is therefore natural that repeating such a sample-then-guess approach ad infinitum will yield an approximately optimal solution. In Theorem \ref{thm:bin_search_general}, we prove this whenever there exist families of progress distributions for every range of possible suboptimalities. Thus, assuming the learner can maintain a (possibly quite pessimistic) estimate of its suboptimality over all the rounds, we obtain a template for proving the iteration complexities of Theorems \ref{cor:bin_search_ip}, \ref{cor:bin_search_pl}, and \ref{cor:bin_search_dist}. Note that $\rho$ can be an arbitrarily small positive constant; even if there is a slim chance of decreasing the suboptimality, this is still sufficient because the monotone adversary ensures that the algorithm can never make negative progress.

\paragraph{Specifying progress distributions.} We now discuss how we instantiate the above template for the $\beta$-smooth (Definition \ref{defn:smoothness}) and $\alpha$-P\L{} (Definition \ref{defn:pl}) case (Theorem \ref{cor:bin_search_pl}). We focus on Theorem \ref{cor:bin_search_pl} for the sake of brevity; the proofs of Theorems \ref{cor:bin_search_ip} and \ref{cor:bin_search_dist} require some additional care but at a high level follow a similar structure. At step $t$, the algorithm maintains a guess $\vx_t$ for the target $\zstar$. It chooses some step size $\eps_t$ and a random vector $\vg_t$ from $\eps_t \cdot \mathsf{Unif}(\S_2^{d-1})$, where $\S_2^{d-1} \coloneqq \inbraces{\vx \in \R^d \suchthat \norm{\vx}_2 = 1}$. We then query $\vx_t$ and $\vx_t - \vg_t$. The key observation is that with a constant probability, the angle between $\vg_t$ and the gradient $\nabla f(\vx)$ is small. We will use this to show that the distribution $\vx_t - \eps_t \cdot \mathsf{Unif}(\S_2^{d-1})$ is a $(1, C_1, C_2)$-progress distribution (Definition \ref{defn:prog}) for constants $C_1, C_2$. Intuitively, this means that $\vx_t - \vg_t$ almost behaves like a step of gradient descent. To turn this observation into an algorithm, we need two main insights.

\paragraph{Step size schedule.} The principal difficulty of this approach is to choose the step size $\eps_t$. It is not immediately obvious how to do so since the algorithm does not observe any actual gradients or function values. Hence, if our step sizes are too large, the algorithm may overshoot the optimal solution $\zstar$ and therefore not actually improve the quality of its current solution $\vx_t$. On the other hand, if our step sizes are too small, the algorithm may not make enough progress in each step, which undesirably increases both the iteration complexity and the total cost. 

To address this, we carefully construct a step size schedule that relies on a pessimistic upper bound on the suboptimality of the algorithm's current solution. With this schedule, we show that in every step, one of two things happens -- either the step size $\eps_t$ is small enough such that there is the possibility of the algorithm decreasing the cost, or it is too large. For the first case, we use $\beta$-smoothness (Definition \ref{defn:smoothness}) to prove that there is a constant probability that the algorithm finds a descent direction, which decreases the cost of its current solution substantially. For the second case, we use the $\alpha$-P\L{} condition (Definition \ref{defn:pl}) to prove that the cost the algorithm incurs in such steps is low. After enough steps, we can show that either the second case always holds (i.e. that the suboptimality is already desirably small) or the maximum cost that the algorithm can pay per round is small. We then decrease the step size $\eps_t$ by a constant factor, update the suboptimality estimate accordingly, and infinitely recurse.

\paragraph{Bounding the failure probability over infinite rounds.} It now remains to show that the probability that the algorithm fails to make enough progress over \textit{infinitely many rounds} is small. This is where the distinction between the two goals of Problem \ref{problem:local_rec_fixed_set} becomes apparent. Specifically, even if we have a subroutine that, with high probability, outputs an $\eps$-approximate solution, this does not immediately convert to an algorithm that can achieve bounded cost over an infinite number of rounds -- note that the failure probabilities may accumulate in a divergent manner. Hence, we will require a more careful probabilistic analysis.

To overcome this challenge, we design the algorithm to run in phases $i = 1, 2, \dots$. In phase $i$, we use a step size $\eps_t$ proportional to $2^{-i/2}$ and run phase $i$ for $\sim id$ steps. Using the fact that the family of distributions we are using for sampling next steps are $(1, C_1, C_2)$-progress distributions, it will be enough to prove that $\sim d \cdot \nfrac{\beta}{\alpha}$ steps yield enough improving steps to decrease the suboptimality by a constant factor. We can therefore apply a Chernoff bound to conclude that the probability that the algorithm fails to make enough progress in phase $i$ is at most $\expv{-id \cdot \nfrac{\beta}{\alpha}}$. Finally, we apply a union bound that the total probability of failure by $\expv{-d \cdot \nfrac{\beta}{\alpha}} \le \expv{-d}$. 

To bound the total cost over all phases $i \in \N_{\ge 1}$, we note that the sum of the suboptimalities in each round is of the form $d\sum_{i \ge 1} i2^{-i} = O(d)$. The guarantee on the iteration complexity follows by noting that to achieve a suboptimality of $2^{-i}$, the algorithm runs $d\sum_{j \le i} i = O\inparen{i^2 \cdot d}$ iterations.

\subsection{Related works}

\paragraph{Dueling convex optimization.} As already mentioned, our formulation in Problem \ref{problem:local_rec_fixed_set} is an extension of dueling convex optimization in the noiseless setting~\cite{jamieson2012query,skm21,saha2022dueling}. \citet{jamieson2012query} employ a coordinate-descent algorithm to show for $\alpha$-smooth and $\beta$-strongly convex $f$,  $\tilde{O}(\nfrac{d\beta}{\alpha} \logv{\nfrac{1}{\eps}})$ queries suffice to learn an $\eps$-optimal point. As mentioned earlier, it is not clear how to adapt their algorithm to handle monotone feedback. In addition, the papers \cite{skm21,saha2022dueling} show results for more general classes of $f$ and in the presence of noise (where the adversary can return invalid response with nonzero probability); see Section \ref{subsec:intro_future_work} for a discussion on extending our results to noisy settings. However, their algorithms explicitly rely on sign feedback $f(\xt{1}) \overset{?}{>} f(\xt{2})$ to construct gradient estimators, which are not possible in the monotone adversary setting.

% This can be converted to an algorithm incurring at most $O(d)$ cost as defined above (ignoring dependencies on smoothness and strong convexity parameters). Note that the setting we study in this paper entails a more challenging feedback model -- in particular, the algorithms of \citet{jamieson2012query} and \citet{skm21} require comparison feedback, whereas in our setting, we do not have this luxury. Hence, our results immediately imply a new algorithm for dueling convex optimization in the noiseless setting where $f$ is $\beta$-smooth and $\alpha$-strongly convex. \glcomment{What about noise?}

\paragraph{Monotone adversaries.} Our setting is an example of learning with a monotone adversary, where an adversary can choose to improve the feedback or information the algorithm gets. A common characteristic is that the improved information may paradoxically break or harm the performance of a given algorithm that works with non-improved information. Monotone adversaries are often studied in the semi-random model literature \cite{blum1995coloring, feige_2021, moitra_2021} for statistical estimation problems \cite{cheng2018non,moitra_2021,kelner2022semi} as well as learning problems, i.e., linear classification with Massart noise \cite{massart2006risk, diakonikolas2019distribution}. 

\paragraph{Preference-based feedback.} Our formulation in this paper falls within the growing body of literature that tackles learning with preference-based feedback, where the algorithm does not learn \emph{how good} its options were in an absolute sense, just which one(s) were better than others. % Examples of settings with quantitative feedback include the (full information) online learning, convex optimization, and reinforcement learning problems, where one typically devises online algorithms that utilize real-valued or vector-valued feedback. However, in many real-world scenarios such as recommendation systems and learning user preferences, one only has access to qualitative comparison-based feedback of the form ``does the user prefer item $\vx$ more than item $\vy$?'' As mentioned, comparison feedback has been studied for dueling bandits and dueling convex optimization; 
Other natural problems with preference-based feedback are contextual search \cite{ls18, lobel2018multidimensional, lls20}, contextual recommendation (also called contextual inverse optimization) \cite{bfl21, gollapudi2021contextual}, and $1$-bit matrix completion \cite{davenport20141}.

\subsection{Future work}
\label{subsec:intro_future_work}

We now discuss two interesting directions our work leaves open.

\paragraph{Dueling optimization with a monotone adversary under noise models.} The most pressing next step is to determine noise models under which we can either obtain algorithmic results or hardness for solving Problem \ref{problem:local_rec_fixed_set}. There are several natural noise models that one could study. As a first step, one can consider the most analogous extension of the noise model studied by \citet{jamieson2012query, skm21}. In the simplest form, they study a noise model where with probability $\nfrac{1}{2} + \nu$ for some parameter $0 < \nu \le \nfrac{1}{2}$, the adversary returns $\argmin_{i \in \inbraces{1,2}} f(\xt{i})$, and with probability $\nfrac{1}{2}-\nu$, the adversary returns $\argmax_{i \in \inbraces{1,2}} f(\xt{i})$. Note that this is a straightforward noise model to handle if the adversary must return one of $\xt{1}$ or $\xt{2}$. The algorithm simply queries the same pair of points roughly $\nu^{-2}$ times, which by Hoeffding's inequality is enough to determine the index $i$ that corresponded to the better of the two guesses. Additionally, this strategy can be implemented even without knowledge of the noise parameter $\nu$; see Section 6 of \cite{skm21} for more details.

The natural extension of this noise model to our monotone feedback setting is as follows. With probability $\nfrac{1}{2}+\nu$, the adversary returns a response $\xstar_t$ that satisfies $f(\xstar_t) \le \min\{f(\xt{1}), f(\xt{2})\}$. On the other hand, with probability $\nfrac{1}{2}-\nu$, the adversary returns an arbitrary point in $\actionset$. It is immediate that there is no analogue of the ``majority vote'' strategy in this setting. % Furthermore, we suspect there exist examples where, for a fixed pair of queries and a given target point, the adversary can cycle between $d$ responses, only $1$ of which is actually worse than any of the queried points (and so this attack will succeed even if $\nfrac{1}{2} - \nu = \nfrac{1}{d}$). Hence, the algorithm cannot tell which $d-1$ of the returned points are actually closer to the target than either of its two guesses. \glcomment{The reviewer will ask about this informal example. It might be best to not draw attention to it.}
We therefore believe that entirely new algorithmic ideas will be needed to address this noise model. On the other hand, it would be very interesting to show an impossibility result for this noise model. An impossibility result would imply that the monotone feedback makes the problem provably harder than the vanilla dueling optimization setting.

Another noise model of interest is one where the adversary's valid monotone feedback is perturbed by a random variable $\triangle$ where $\exv{\triangle} = 0$ and $\mathsf{Cov}\insquare{\triangle} \preceq \sigma^2 \mI_d$. To our knowledge, this noise model has not been considered in past works dealing with optimization with preference-based feedback, even in the vanilla dueling setting (particularly the works of \citet{jamieson2012query} and \citet{skm21}).

\paragraph{Dueling optimization for other function classes.} An orthogonal thread would be to identify other function classes and feasible regions $\actionset$ for which we can build algorithms for Problem \ref{problem:local_rec_fixed_set}. For instance, \citet{skm21} obtain results for functions that are just $\beta$-smooth (and not necessarily $\alpha$-P\L{} or $\alpha$-strongly convex). Our Theorem \ref{cor:bin_search_pl} uses the $\alpha$-P\L{} condition to prove that steps that do not make much progress also do not incur much cost. Hence, to remove the $\alpha$-P\L{} assumption, one would need to either avoid this line of reasoning or find another way to argue that the costs in low-progress rounds are not too high.

\section{Proofs of upper bound results}

In this section, we prove Theorem~\ref{thm:bin_search_general} (in which we construct and analyze a meta-algorithm for Problem \ref{problem:local_rec_fixed_set} when the algorithm can sample next steps from progress distributions (Definition \ref{defn:prog})). We then show how to use this framework to prove  Theorem~\ref{cor:bin_search_ip} (results for $f(\vx) = \ip{-\vx, \zstar}$), Theorem~\ref{cor:bin_search_pl} (results for $f(\vx)$ being $\beta$-smooth and $\alpha$-P\L{}), and Theorem~\ref{cor:bin_search_dist} (results for $f(\vx) = \norm{\vx - \zstar}_2$), in that order. It will be helpful to recall the overview from Section \ref{sec:tech-overview} throughout this section.

We prove Theorem~\ref{thm:bin_search_general} in Section \ref{subsec:general_proof}, Theorem~\ref{cor:bin_search_ip} in Section \ref{subsec:ip_proof}, Theorem~\ref{cor:bin_search_pl} in Section \ref{subsec:pl_proof}, and Theorem~\ref{cor:bin_search_dist} in Section \ref{subsec:dist_proof}.

Before we jump into the main proofs, we will need a couple straightforward numerical inequalities.

\begin{lemma}
\label{lemma:cost_convergence_series}
For $r \in (0, 1)$ and $1 \le p < 2$, we have $\sum_{i \ge 0} i \cdot r^{(1-p/2)i} \le \frac{r^{p/2+1}}{(r-r^{p/2})^2}$.
\end{lemma}
\begin{proof}[Proof of Lemma \ref{lemma:cost_convergence_series}]
Recall that
\begin{align*}
    \sum_{i \ge 0} r^{(1-p/2)i} = \frac{1}{1-r^{1-p/2}}.
\end{align*}
Taking the derivative of both sides with respect to $r$ yields
\begin{align*}
    \sum_{i \ge 0} \inparen{1-\nfrac{p}{2}}i\cdot r^{(1-p/2)i - 1} = \frac{(2-p)r^{p/2}}{2(r-r^{p/2})^2}.
\end{align*}
We multiply both sides by $r$ and divide both sides by $1-\nfrac{p}{2}$; we conclude that
\begin{align*}
    \sum_{i \ge 0} i \cdot r^{(1-p/2)i} = \frac{r^{p/2+1}}{(r-r^{p/2})^2}
\end{align*}
which recovers the statement of Lemma \ref{lemma:cost_convergence_series}.
\end{proof}

\begin{lemma}[Inner product with a random vector]\label{lem:random-progress}
Let $\vg \sim \mathsf{Unif}(\S_2^{d-1})$ and let $\vy \in \S_2^{d-1}$ be fixed. Then
\begin{align*}
    \prvv{\vg}{\ip{\vg,\vy}  \ge \frac{1}{2\sqrt{d}}} \ge \frac{1}{8}.
\end{align*}
\end{lemma}
\begin{proof}
By rotational invariance, without loss of generality, we can let $\vy = \ve_1$. We apply Lemma 2.2 (a) due to \citet{dasgupta} with $\beta = \nfrac{1}{4}$ to conclude that
\begin{align*}
    \prvv{\vg}{\vg_1^2 \le \frac{1}{4d}} \le \expv{\frac{1}{2}\inparen{1-\frac{1}{4}+\ln\inparen{\frac{1}{4}}}} < \frac{3}{4}
\end{align*}
which means that $ \prvv{\vg}{\abs{\ip{\vg,\vy}}  \ge \frac{1}{2\sqrt{d}}} \ge \frac{1}{4}$. The result of Lemma \ref{lem:random-progress} follows by symmetry.
\end{proof}

\subsection{A general algorithm for Problem \ref{problem:local_rec_fixed_set} with progress distributions}
\label{subsec:general_proof}

The goal of this subsection is to develop the general tools we need to prove our main results.

The key primitive of our analysis is a general algorithm (Algorithm \ref{alg:main_alg}) that solves Problem \ref{problem:local_rec_fixed_set} when we are given certain convenient distributions from which we sample new guesses. We call these \textit{progress distributions}; recall Definition \ref{defn:prog}.

Let us describe Algorithm \ref{alg:main_alg}. In each step, Algorithm \ref{alg:main_alg} maintains a current guess $\vx_t$ and chooses a slight perturbation of that guess $\xplus_t \sim \cD(\vx_t)$, where $\cD(\vx_t)$ is a $(p, \gamma, \rho)$-progress distribution (Definition \ref{defn:prog}). Algorithm \ref{alg:main_alg} then submits the pair of guesses $\inbraces{\vx_t, \xplus_t}$. To analyze Algorithm \ref{alg:main_alg}, the main observation is that with probability $\ge \gamma$, the point $\xplus_t$ substantially improves over the cost of $\vx_t$ -- this follows directly from Definition \ref{defn:prog}. We exploit this intuition to give our most general result (Theorem \ref{thm:bin_search_general}) and to prove the correctness of Algorithm \ref{alg:main_alg}.

\begin{theorem}
\label{thm:bin_search_general}
Let $\myfunc{f}{\actionset}{\R}$. Let $B$ and $\vx_1 \in \actionset$ be such that $f(\vx_1)-f(\zstar) \le B$. For $C > 0$, constant $r \in (0,0.99)$, and for all $i \in \N_{\ge 1}$, suppose there exists intervals of the form $C \cdot \insquare{r^{i+1}, r^{i}}$ such that their union covers the interval $[0, B]$.

If there exists a $(p, \gamma,\rho)$-progress distribution $\cD_i(\vx)$ whenever $f(\vx)-f(\zstar) \in C \cdot \insquare{r^{i+1}, r^{i}}$ for all $i \ge 1$ and where $p, \gamma, \rho$ do not depend on $\vx$ and $i$, then there is an algorithm (Algorithm \ref{alg:main_alg}) for Problem \ref{problem:local_rec_fixed_set} that, with probability at least $1 - \expv{-O\inparen{\frac{d}{\rho B^{p-1}}}}$, incurs total cost
\begin{align*}
    O\inparen{\frac{B\logv{\nfrac{1}{r}}}{B^{p-1}\gamma\rho\min\inbraces{r^{p\inparen{\nfrac{p-1}{2-p}}}, \inparen{r-r^{p/2}}^2}} \cdot d}.
\end{align*}
Additionally, Algorithm \ref{alg:main_alg} finds a point $\vx$ satisfying $f(\vx)-f(\zstar) \le \eps$ in 
\begin{align*}
    O\inparen{\frac{1}{B^{p-1}\gamma\rho} \cdot d \cdot \logv{\frac{B}{\eps}}^2}
\end{align*}
iterations with at least the aforementioned probability.
\end{theorem}

\begin{algorithm}%[H]
\caption{General recipe algorithm for dueling convex optimization}\label{alg:main_alg}
\begin{algorithmic}[1]
    \STATE \textbf{Input}: Interaction with a monotone adversary $\cM$ as defined in Problem \ref{problem:local_rec_fixed_set}; initial point $\vx_1$ and bound $B$ satisfying $f(\vx_1)-f(\zstar) \le B$; values $C$ and $r$ for which there exist corresponding intervals and $(p,\gamma, \rho)$-progress distribution families $\cD_i$ (see the statement of Theorem \ref{thm:bin_search_general}).
    \STATE Initialize $\vx_1 = 0$, $t=1$.
    \FOR{$i = 1, \dots$}
        \FOR{$T(i) \coloneqq \nfrac{2i}{(\gamma\min(1,\rho))} \cdot (Cr^{i+1})^{-(p-1)} \cdot  \logv{\nfrac{1}{r}} \cdot d$ iterations}\label{line:main_alg_ti}
            \STATE Sample $\xplus_t$ from $\cD_i(\vx_t)$.
            \STATE \textbf{Submit guesses} $\inbraces{\vx_t, \xplus_t}$ and  \textbf{receive response} $\xstar_t$.
            \STATE Let $\vx_{t+1} = \xstar_t$.
            \STATE Update $t \gets t+1$.
        \ENDFOR
    \ENDFOR
\end{algorithmic}
\end{algorithm}

The proof of Theorem \ref{thm:bin_search_general} has two main parts. In the first part, we will prove that for each value of $i$ (call the set of timesteps belonging to a particular value of $i$ ``phase $i$''), the number of steps $T(i)$ is sufficient to ensure that the cost of the algorithm's solution decays gracefully with sufficiently large probability. In the second part, we will prove that the total cost the algorithm pays over all phases $i \ge 1$ is $\sim B \cdot \nfrac{d}{\gamma\rho}$ as promised. Theorem \ref{thm:bin_search_general} will easily follow by combining these facts.

We start with stating and proving Lemma \ref{lemma:steps_per_phase}.

\begin{lemma}
\label{lemma:steps_per_phase}
Let $i \ge 1$ and let $T(i)$ be defined below (or see Line \ref{line:main_alg_ti} of Algorithm \ref{alg:main_alg}).
\begin{align*}
    T(i) \coloneqq \frac{2i}{\gamma\min(1,\rho)\cdot (Cr^{i+1})^{p-1}}  \cdot  \logv{\nfrac{1}{r}} \cdot d.
\end{align*}
Let $t_i$ be the first iteration of phase $i$. If $f(\vx_{t_i}) - f(\zstar) \le Cr^i$, then, with probability $\ge 1 - r^{\frac{dB^{-(p-1)}}{4\rho} \cdot i}$, we have $f(\vx_{t_i+T(i)+1}) - f(\zstar) \le Cr^{i+1}$.
\end{lemma}
\begin{proof}[Proof of Lemma \ref{lemma:steps_per_phase}]
Assume that we have $f(\vx_{t_i})-f(\zstar) \in C \cdot \insquare{r^{i+1}, r^i}$ (otherwise, we are done immediately).

Define the indicator random variable $Y_t$ as follows.
\begin{align*}
    Y_t \coloneqq \indicator{\frac{f(\vx_{t})-f(\xplus_t)}{\inparen{f(\vx_{t})-f(\zstar)}^p} \ge \frac{\rho}{d}}.
\end{align*}

Consider the distribution of guesses $\cD_i$ (let us omit the argument $\vx_t$ for the sake of brevity). Since $\cD_i$ is a $(p,\gamma,\rho)$-progress distribution, we have
\begin{align*}
    \prvv{\xplus \sim \cD_i}{\frac{f(\vx_t)-f(\xplus_t)}{f(\vx_t)-f(\zstar)} \ge \frac{\rho}{d} \cdot \inparen{Cr^{i+1}}^{p-1}} \ge \prvv{\xplus \sim \cD_i}{Y_t = 1}  \ge \gamma.
\end{align*}
Call every step $t$ for which $Y_t = 1$ a ``successful step.'' Let us give a high-probability count on the number of successful steps. Recall that a form of the Chernoff bound states that, for $\delta \in [0,1]$ and independent indicator random variables $Y_j$,
\begin{align*}
    \prv{\sum_{j = t_i}^{t_i+T(i)} Y_j \le (1-\delta)\exv{\sum_{j=t_i}^{t_i+T(i)} Y_j}} \le \expv{-\frac{\delta^2 \cdot \exv{\sum_{j=t_i}^{t_i+T(i)} Y_j}}{2}}.
\end{align*}
Applying the Chernoff bound with $\delta = \nfrac{1}{2}$ yields
\begin{align*}
    \prv{\sum_{j=t_i}^{t_i+T(i)} Y_j \le \frac{T(i)\gamma}{2}} \le \expv{-\frac{i \cdot \frac{d}{\rho\inparen{Cr^{i+1}}^{p-1}} \cdot \logv{\nfrac{1}{r}}}{4}} \le r^{\frac{dB^{-(p-1)}}{4\rho} \cdot i}
\end{align*}
where we use $Cr^{i+1} \le Cr^2 \le B$.

It remains to show that after at least $\nfrac{T(i)\gamma}{2}$ successful steps, we have $f(\vx_{t_i + T(i) + 1}) - f(\zstar) \le Cr^{i+1}$. Recall that we assume that $f(\vx_{t_i}) - f(\zstar) \ge Cr^{i+1}$ and note that for every successful step, we have
\begin{align*}
    \frac{f(\vx_t)-f(\xplus_t)}{f(\vx_t)-f(\zstar)} \ge \frac{\rho}{d} \cdot \inparen{Cr^{i+1}}^{p-1}
\end{align*}
which implies
\begin{align*}
    \frac{f(\xplus_t)-f(\zstar)}{f(\vx_t)-f(\zstar)} \le 1 - \frac{\rho}{d} \cdot \inparen{Cr^{i+1}}^{p-1}.
\end{align*}
We multiply over all steps in phase $i$, giving
\begin{align*}
    \frac{f(\vx_{t_i + T(i) + 1})-f(\zstar)}{f(\vx_{t_i})-f(\zstar)} &= \prod_{t = t_i}^{t_i + T(i)} \frac{f(\vx_{t+1}) - f(\zstar)}{f(\vx_t) - f(\zstar)} \le \inparen{1 - \frac{\rho}{d} \cdot \inparen{Cr^{i+1}}^{p-1}}^{T(i)\gamma/2} \\
    &\le \inparen{1 - \frac{2i}{\gamma} \cdot \frac{\logv{\nfrac{1}{r}}}{T(i)}}^{T(i)\gamma/2} \le \expv{\frac{2i}{\gamma} \cdot \frac{\logv{\nfrac{1}{r}}}{T(i)} \cdot \frac{T(i)\gamma}{2}} = r^{i} \le r.
\end{align*}
Finally, recall that $f(\vx_{t_i}) - f(\zstar) \le Cr^i$. Combining this with the above gives $f(\vx_{t_i + T(i) + 1}) - f(\zstar) \le Cr^{i+1}$, concluding the proof of Lemma \ref{lemma:steps_per_phase}.
\end{proof}

% We now take a brief interlude to prove a straightforward numerical inequality that we will need for the rest of the analysis. See Lemma \ref{lemma:cost_convergence_series}.
Next, we have Lemma \ref{lemma:total_cost}, which controls the total cost that Algorithm \ref{alg:main_alg} incurs assuming that the cost is sufficiently low in each phase.

\begin{lemma}
\label{lemma:total_cost}
For a timestep $t$, let $i(t)$ be the phase that $t$ belongs to.

If for all $t$ we have $f(\vx_t) - f(\zstar) \le Cr^{i(t)}$, then Algorithm \ref{alg:main_alg} incurs total cost
\begin{align*}
    O\inparen{\frac{B\logv{\nfrac{1}{r}}}{C^{p-1}\gamma\rho\min\inbraces{r^{p\inparen{\nfrac{p-1}{2-p}}}, \inparen{r-r^{p/2}}^2}} \cdot d}.
\end{align*}
\end{lemma}
\begin{proof}[Proof of Lemma \ref{lemma:total_cost}]
Recall throughout this proof that $r \le 0.99$ and $p$ is a constant such that $p < 2$.

Observe that in phase $i$, the algorithm incurs cost at most
\begin{align*}
    T(i) \cdot Cr^i = \frac{2i}{\gamma} \cdot \frac{dCr^i}{\rho\inparen{Cr^{i+1}}^{p-1}} \cdot \logv{\nfrac{1}{r}} = \frac{2i}{\gamma} \cdot \frac{d\logv{\nfrac{1}{r}}}{\rho C^{p-2}} \cdot r^{i - (i+1)(p-1)}.
\end{align*}
We will find a threshold $i_p$ for which for all $i \ge i_p$, the above cost is exponentially decaying. This will allow us to control the sum of the costs over infinitely many rounds. We choose $i_p = 2 \cdot \ceil{\inparen{\nfrac{(p-1)}{(2-p)}}}$. Notice that for all $i \ge i_p$, the exponent on $r$ can be bounded as
\begin{align*}
    i-(i+1)(p-1) = i(2-p) - (p-1) \ge \inparen{1-\frac{p}{2}}i.
\end{align*}
Note that this also implies that $(i_p+1)(p-1) \le \nfrac{p}{2} \cdot i_p = p\ceil{\nfrac{(p-1)}{(2-p)}}$.

To total the cost, we consider two cases. First, suppose $1 \le i \le i_p - 1$. Observe that in each of these phases, we pay cost at most $B$, so we have
\begin{align*}
    \sum_{i=1}^{i_p-1} B\cdot T(i) &\le B \inparen{T_{i_p} \cdot i_p} = 2B\inparen{2 \cdot \frac{p-1}{2-p}}^2 \cdot \frac{\logv{\nfrac{1}{r}}}{\gamma} \cdot \frac{d}{\rho C^{p-1}r^{(i_p+1)(p-1)}} \\
    &\le 2B\inparen{2 \cdot \frac{p-1}{2-p}}^2 \cdot \frac{\logv{\nfrac{1}{r}}}{\gamma} \cdot \frac{d}{\rho C^{p-1}r^{p\inparen{\nfrac{p-1}{2-p}}}}. \numberthis\label{eq:bound1}
\end{align*}
Next, we sum over all phases $i \ge i_p$. We obtain a cost that is at most
\begin{align*}
    \sum_{i \ge i_p} \frac{2i}{\gamma} \cdot \frac{d\logv{\nfrac{1}{r}}}{\rho C^{p-2}} \cdot r^{i - (i+1)(p-1)} \le \frac{2d\logv{\nfrac{1}{r}}}{\gamma\rho \cdot C^{p-2}}\sum_{i \ge i_p} i \cdot r^{(1-p/2)i} \le \frac{2d\logv{\nfrac{1}{r}}}{\gamma\rho \cdot C^{p-2}} \cdot \frac{r^{p/2+1}}{(r-r^{p/2})^2} \numberthis\label{eq:bound2}
\end{align*}
where the last inequality follows from Lemma \ref{lemma:cost_convergence_series}. Combining \eqref{eq:bound1} and \eqref{eq:bound2} yields
\begin{align*}
    &\quad 2B\inparen{2 \cdot \frac{p-1}{2-p}}^2 \cdot \frac{\logv{\nfrac{1}{r}}}{\gamma} \cdot \frac{d}{\rho C^{p-1}r^{p\inparen{\nfrac{p-1}{2-p}}}}+\frac{2d\logv{\nfrac{1}{r}}}{\gamma\rho \cdot C^{p-2}} \cdot \frac{r^{p/2+1}}{(r-r^{p/2})^2} \\
    &\le 2B\inparen{2 \cdot \frac{p-1}{2-p}}^2 \cdot \frac{\logv{\nfrac{1}{r}}}{\gamma} \cdot \frac{d}{\rho C^{p-1}r^{p\inparen{\nfrac{p-1}{2-p}}}}+\frac{2d\logv{\nfrac{1}{r}}}{\gamma\rho \cdot C^{p-1}} \cdot \frac{B}{(r-r^{p/2})^2} \\
    &= O\inparen{\frac{B\logv{\nfrac{1}{r}}}{C^{p-1}\gamma\rho\min\inbraces{r^{p\inparen{\nfrac{p-1}{2-p}}}, \inparen{r-r^{p/2}}^2}} \cdot d}
\end{align*}
This concludes the proof of Lemma \ref{lemma:total_cost}.
\end{proof}

We are now ready to prove Theorem \ref{thm:bin_search_general}.

\begin{proof}[Proof of Theorem \ref{thm:bin_search_general}]
It is sufficient to prove that with probability $\ge 1 - \expv{-O\inparen{\frac{d}{\rho B^{p-1}}}}$, at the end of phase $i$, we have $f(\vx_t) - f(\zstar) \le Cr^{i+1}$. Recall the conclusion of Lemma \ref{lemma:steps_per_phase} and that $f(\vx_1) - f(\zstar) \le B \le Cr$; by a union bound, we have for all phases $i$ that $f(\vx_t) - f(\zstar) \le Cr^{i+1}$ with probability
\begin{align*}
    1 - \sum_{i \ge 1} r^{\frac{dB^{-(p-1)}}{4\rho} \cdot i} \ge 1 - \expv{-O\inparen{\frac{d}{\rho B^{p-1}}}}
\end{align*}
where we use $0 < r < 0.99$. The first part of Theorem \ref{thm:bin_search_general} now follows directly from applying Lemma \ref{lemma:total_cost}. The rest of the statement of Theorem \ref{thm:bin_search_general} follows by noting that
\begin{align*}
    \sum_{j \le i} T(j) = \frac{2C^{-(p-1)}\logv{\nfrac{1}{r}}}{\gamma\min(1,\rho)} \cdot d \cdot \sum_{j \le i} j\inparen{r^{-(j+1)(p-1)}} \lesssim \frac{2C^{-(p-1)}}{\gamma\min(1,\rho)} \cdot d \cdot i^2.
\end{align*}
where we again use $r < 0.99$. We set $\eps = Cr^i$ and conclude.
\end{proof}

\subsection{Proof of Theorem \ref{cor:bin_search_ip}}
\label{subsec:ip_proof}

The goal of this subsection is to prove Theorem \ref{cor:bin_search_ip}.

Our plan will be to use the general guarantee of Theorem \ref{thm:bin_search_general}. Thus, the main task is to prove that there is an appropriate interval cover and corresponding sequence $\cD_i(\vx)$ of progress distributions for all $\vx$ belonging to phase $i$ that satisfy the conditions of Theorem \ref{thm:bin_search_general}.

We prove this fact in Lemma \ref{lem:inner-prod-dists}. We remark that we made no effort to optimize the numerical constants; we choose the constants that appear in the Lemma statement to simplify calculations, as these will not impact our asymptotic results.

\begin{lemma}
\label{lem:inner-prod-dists}
Let $\myfunc{f}{\R^d}{\R}$ be the negative inner product function defined on $\mathbb{S}_2^{d-1}$ with respect to some unknown target $\zstar$. Then for any $\vx$ for which $f(\vx)-f(\zstar) \in \insquare{10^{-(i+1)}, 10^{-i}}$ and for which $\langle \zstar,\vx\rangle>0$, there is a $(1.5,10^{-1},10^{-4})$-progress distribution (Definition \ref{defn:prog}) that can be computed in time $O(d)$.
\end{lemma}

\begin{proof}[Proof of Lemma \ref{lem:inner-prod-dists}]
We explain the construction of the distribution $\cD(\vx)$.

Impose the following coordinates on $\R^d$. Let the first coordinate $x_1$ be the direction of $\vx$, and the remaining $d-1$ coordinates be an arbitrary coordinate system for the perpendicular directions. Then, $\vx$ has coordinates $(1,0\ldots)$. Next, let $z \coloneqq 10^{-(i+1)}$ and $s:=\frac{z}{10\sqrt{d-1}}$. Let $\vs$ be a point randomly drawn from a $d-1$ dimensional sphere of radius $s$ whose coordinates are denoted $s_1\ldots s_{d-1}$. Then, the distribution $\cD\inparen{\vx}$ is the distribution of $(\sqrt{1-s^2},s_1\ldots s_{d-1})$. It is easy to verify that these points lie on $\mathbb{S}_2^{d-1}$.

This distribution can be computed in time $O(d)$. We will now show that it is a $(1.5,10^{-1},10^{-4})$ progress distribution.

Let $z_0 \coloneqq f(\vx)-f(\zstar)$; recall that $z_0 \in [0,1]$ and $z \le z_0 \le 10z$. We write the target vector $\zstar=(1-z_0)\vx+\sqrt{1-(1-z_0)^2}\vy = (1-z_0)\vx+\sqrt{2z_0-z_0^2}\vy$ where $\vy$ is a unit vector and $\ip{\vy,\vx}=0$. Note that this expression holds because $\langle \zstar, \vx\rangle = 1-z_0$.

Let $\xplus$ be a random point chosen from $\cD(\vx)$. Let $y_1\ldots y_{d-1}$ be the coordinates of $\vy$ in the $d-1$ dimensional coordinate plane perpendicular to $\vx$ defined above. We compute
\begin{align} \label{eq:xxplus}
    \langle \zstar, \xplus \rangle
    =&~\inparen{1-z_0,\sqrt{2z_0-z_0^2}y_1\ldots \sqrt{2z_0-z_0^2}y_{d-1}}\cdot \inparen{\sqrt{1-s^2},s_1\ldots s_{d-1}} \\
    =&~(1-z_0)\sqrt{1-s^2}+\inparen{\sqrt{2z_0-z_0^2}y_1\ldots \sqrt{2z_0-z_0^2}y_{d-1}} \cdot (s_1\ldots s_{d-1}).
\end{align}
By Lemma~\ref{lem:random-progress}, we have (note that we weaken the constants from Lemma \ref{lem:random-progress} for numerical convenience later in the proof)
\begin{align*}
    \prvv{\vs}{\inparen{\sqrt{2z_0-z_0^2}y_1\ldots \sqrt{2z_0-z_0^2}y_{d-1}} \cdot (s_1\ldots s_{d-1}) \geq \frac{0.1}{\sqrt{d-1}}\cdot \sqrt{2z_0-z_0^2} \cdot s} \ge 0.1.
\end{align*}
Because $2z_0-z_0^2\geq z_0\geq z$, we have
\begin{align*}
     \frac{0.1}{\sqrt{d-1}}\cdot \sqrt{2z_0-z_0^2} \cdot s \geq \frac{s\sqrt{z}}{10\sqrt{d-1}}.
\end{align*}
In turn, this shows
\begin{align*}
    \prvv{s}{\inparen{\sqrt{2z_0-z_0^2}y_1\ldots \sqrt{2z_0-z_0^2}y_{d-1}} \cdot (s_1\ldots s_{d-1}) \geq \frac{s\sqrt{z}}{10\sqrt{d-1}}} \ge 0.1.
\end{align*}
Combining this with Equation~\ref{eq:xxplus}, we obtain
\begin{align} \label{eq:xstarxplus}
    \prvv{\xplus \sim \cD_{\vx}}{\langle \zstar, \xplus \rangle \ge (1-z_0)\sqrt{1-s^2}+\frac{s\sqrt{z}}{10\sqrt{d-1}}} \ge 0.1.
\end{align}
Now, we will find a lower bound for $(1-z_0)\sqrt{1-s^2}+\frac{s\sqrt{z}}{10\sqrt{d-1}}$. We have
\begin{align*}
    (1-z_0)\sqrt{1-s^2}+\frac{s\sqrt{z}}{10\sqrt{d-1}} &\geq (1-z_0)(1-s^2)+\frac{s\sqrt{z}}{10\sqrt{d-1}} \\
    &= (1-z_0)(-s^2)+\frac{s\sqrt{z}}{10\sqrt{d-1}} + (1-z_0) \\
    &\geq (1-z)(-s^2)+\frac{s\sqrt{z}}{10\sqrt{d-1}} + (1-z_0).
\end{align*}
Now, using that $s = \frac{z}{10\sqrt{d-1}}$, we get 
\begin{align*}
    (1-z)(-s^2)+\frac{s\sqrt{z}}{10\sqrt{d-1}} + (1-z_0)
    &= (1-10s\sqrt{d-1})(-s^2)+\frac{s^{3/2}}{\sqrt{10}(d-1)^{1/4}} + (1-z_0)\\
    &= -s^2+10s^3\sqrt{d-1}+\frac{s^{3/2}}{\sqrt{10}(d-1)^{1/4}} + (1-z_0) \\
    &\geq \left(10s^3\sqrt{d-1}+\frac{s^{3/2}}{8(d-1)^{1/4}}-s^2\right) \\
    &\quad\quad+\frac{s^{3/2}}{6(d-1)^{1/4}} + (1-z_0)
\end{align*}
where the last line follows from $\nfrac{1}{\sqrt{10}} > \nfrac{1}{8}+\nfrac{1}{6}$. Finally, applying weighted AM-GM lets us see
\begin{align*}
    10s^3\sqrt{d-1}+\frac{s^{3/2}}{8(d-1)^{1/4}}\geq \frac{3}{2^{2/3}}\inparen{10s^3\sqrt{d-1}}^{1/3}\inparen{\frac{s^{3/2}}{8(d-1)^{1/4}}}^{2/3} = \frac{3\cdot 10^{1/3}}{2^{2/3}2^{2}}s^2 > s^2
\end{align*}
where we use a weight of $\nfrac{1}{3}$ on the first term and a weight of $\nfrac{2}{3}$ on the second term. We now write
\begin{align*}
    (1-z_0)\sqrt{1-s^2}+\frac{s\sqrt{z}}{10\sqrt{d-1}} \geq \frac{s^{3/2}}{6(d-1)^{1/4}} + (1-z_0).
\end{align*}
Substituting $s$ once again and recalling that $\langle \zstar, \vx\rangle = (1-z_0)$ and $z \ge \frac{z_0}{10} = \frac{\langle \zstar, \zstar - \vx\rangle}{10}$, we get
\begin{align*}
    (1-z_0)\sqrt{1-s^2}+\frac{s\sqrt{z}}{10\sqrt{d-1}} \geq \frac{z^{3/2}}{6\cdot 10^{3/2} \cdot d} + \langle \zstar, \vx \rangle > \frac{10^{-4}\langle \zstar, \zstar - \vx\rangle^{3/2}}{d} + \langle \zstar, \vx \rangle.
\end{align*}
Combining this with \eqref{eq:xstarxplus}, we now have
\begin{align*}
    &\prvv{\xplus \sim \cD\inparen{\vx}}{\langle \zstar, \xplus \rangle \ge \frac{10^{-4}\langle \zstar, \zstar - \vx\rangle^{3/2}}{d} + \langle \zstar, \vx\rangle} \ge 0.1
\end{align*}
which means that
\begin{align*}
    \prvv{\xplus \sim \cD\inparen{\vx}}{\langle \zstar, \xplus - \vx \rangle \ge \frac{10^{-4}\langle \zstar, \zstar - \vx\rangle^{3/2}}{d}} \ge 0.1.
\end{align*}
This exactly aligns with the definition of a $(1.5, 10^{-1}, 10^{-4})$ progress distribution, completing the proof of Lemma \ref{lem:inner-prod-dists}.
\end{proof}

We will now conclude Theorem~\ref{cor:bin_search_ip} using Theorem~\ref{thm:bin_search_general}.

\begin{proof}[Proof of Theorem~\ref{cor:bin_search_ip}]
To apply Theorem \ref{thm:bin_search_general}, we need to present $C$, $r$, and a sequence of parameterizations $\cD_i$ that satisfy the premises.

Set $r=0.1$ and $C=10$. By Lemma~\ref{lem:inner-prod-dists}, we can find progress distributions for each interval $[Cr^i,Cr^{i-1}]$ of suboptimality of the current function value, since we can find such progress distributions as long as the suboptimality of the function is at most $1$. 

Note that the algorithm can begin with a point $\vx$ where $f(\vx)-f(\zstar)<1$ by first querying two opposite points on a sphere; one can easily see that at least one of the two points queried satisfies $f(\vx)-f(\zstar)<1$.

We therefore conclude the proof of Theorem \ref{cor:bin_search_ip}.
\end{proof}

\subsection{Proof of Theorem \ref{cor:bin_search_pl}}
\label{subsec:pl_proof}

The goal of this subsection is to prove Theorem \ref{cor:bin_search_pl}.

As before, we use the general guarantee of Theorem \ref{thm:bin_search_general} via proving that there is an appropriate interval cover and corresponding sequence $\cD_i(\vx)$ of progress distributions for all $\vx$ for which $f(\vx)-f(\zstar) \in \insquare{Cr^{i+1},Cr^i}$.

We prove this fact in Lemma \ref{lemma:pl_prog_distro}.

\begin{lemma}
\label{lemma:pl_prog_distro}
Fix $i \in \N_{\ge 1}$. Let $\eps_i = \sqrt{\nfrac{2B \alpha}{\beta^2}} \cdot \nfrac{1}{\sqrt{d}}\cdot 2^{-i/2-1}$. If $f$ is $\beta$-smooth and $\alpha$-P\L{}, and if we have $f(\vx) - f(\zstar) \in B \cdot \insquare{2^{-i}, 2^{-i+1}}$, then the distribution $\cD_i(\vx) = \vx + \eps_i \cdot \mathsf{Unif}(\S_2^{d-1})$ is a $(1,\gamma,\rho)$-progress distribution for $(\gamma,\rho) = \inparen{\nfrac{1}{8}, \nfrac{\alpha}{8\beta}}$.
% If $f$ is $\beta$-smooth and $\alpha$-P\L{}, and if we have $f(\vx) - f(\zstar) \in \max(2B,B \cdot \nfrac{\beta^2}{\alpha}) \cdot \insquare{2^{-i-1}, 2^{-i}}$
\end{lemma}
\begin{proof}[Proof of Lemma \ref{lemma:pl_prog_distro}]
Let $\vg \coloneqq \vx - \xplus$.

It is sufficient to consider the case where we have $\norm{\vg}_2 \le \frac{1}{2\beta} \cdot \frac{\norm{\nabla f(\vx)}_2}{\sqrt{d}}$. To see this, suppose this is not the case. We apply the P\L{} inequality and write
\begin{align*}
    f(\vx) - f(\zstar) \le \frac{1}{2\alpha}\norm{\nabla f(\vx)}_2^2 \le d \cdot \frac{2\beta^2}{\alpha}\norm{\vg}_2^2 = d \cdot \frac{2\beta^2}{\alpha} \inparen{ \sqrt{\frac{2B \alpha}{\beta^2 d}} \cdot \frac{1}{2^{i/2+1}}}^2 = \frac{B}{2^{i}}
\end{align*}
which implies that the suboptimality $f(\vx) - f(\zstar)$ does not belong to the range we are considering.  

Next, we use Lemma \ref{lem:random-progress} to write the below.
\begin{align*}
    \prvv{\vg}{\ip{\frac{\nabla f(\vx)}{\norm{\nabla f(\vx)}_2},\frac{\vg}{\norm{\vg}_2}} \ge \frac{1}{2\sqrt{d}}} \ge \frac{1}{8}.
\end{align*}
By Definition \ref{defn:smoothness}, we have for a $\beta$-smooth function and for any $\vx, \vy \in \R^d$ that
\begin{align*}
    \abs{f(\vx)-f(\vy) - \ip{\nabla f(\vy), \vx-\vy}} \le \frac{\beta}{2} \cdot \norm{\vx-\vy}_2^2,
\end{align*}
from which it easily follows that
\begin{align*}
    \abs{f(\vx-\vg)-f(\vx)+\ip{\nabla f(\vx), \vg}} \le \frac{\beta}{2} \cdot \norm{\vg}_2^2.
\end{align*}
The above rearranges to
\begin{align*}
    f(\vx) - f(\vx-\vg) &\ge \ip{\nabla f(\vx),\vg} - \frac{\beta}{2} \cdot \norm{\vg}_2^2 \\
    &= \norm{\nabla f(\vx)}_2 \cdot \norm{\vg}_2 \cdot \inparen{\ip{\frac{\nabla f(\vx)}{\norm{\nabla f(\vx)}_2},\frac{\vg}{\norm{\vg}_2}} - \frac{\nfrac{\beta}{2} \cdot \norm{\vg}_2}{\norm{\nabla f(\vx)}_2}} \\
    &\ge \frac{\beta}{2} \cdot \norm{\vg}_2^2\quad\quad\text{with probability } > \nfrac{1}{8}
\end{align*}
We therefore conclude that with probability $> \gamma \coloneqq \nfrac{1}{8}$,
\begin{align*}
    f(\vx)-f(\vx-\vg)\ge \frac{\beta}{2} \cdot \norm{\vg}_2^2 = \frac{\beta}{2} \cdot \inparen{ \sqrt{\frac{2B \alpha}{\beta^2 d}} \cdot \frac{1}{2^{i/2+1}}}^2 = \frac{\alpha}{\beta} \cdot \frac{B}{d} \cdot \frac{1}{2^{i+2}}.
\end{align*}
This means that
\begin{align*}
    \frac{f(\vx)-f(\vx-\vg)}{f(\vx) - f(\zstar)} \ge \frac{
     \nfrac{\alpha}{\beta} \cdot \nfrac{B}{d} \cdot \nfrac{1}{2^{i+2}}
    }{\nfrac{B}{2^{i-1}}} = \frac{\alpha}{\beta} \cdot \frac{1}{8d}
\end{align*}
which means we can take $\rho = \nfrac{\alpha}{8\beta}$. This concludes the proof of Lemma \ref{lemma:pl_prog_distro}.
\end{proof}

The proof of Theorem \ref{cor:bin_search_pl} follows very easily from Lemma \ref{lemma:pl_prog_distro}.

\begin{proof}[Proof of Theorem \ref{cor:bin_search_pl}]
Our plan is to apply Theorem \ref{thm:bin_search_general}. To do so, we need to present $C$, $r$, and a sequence of $\cD_i$ that satisfy the premise of Theorem \ref{thm:bin_search_general}. We will use the settings of these objects guaranteed by Lemma \ref{lemma:pl_prog_distro}.

Let $C = 2B$ and $r = \nfrac{1}{2}$. It is clear that the intervals given by Lemma \ref{lemma:pl_prog_distro} cover $[0,B]$, and so for every $i \ge 1$, there exists a corresponding $\inparen{1,\nfrac{1}{8},\nfrac{\alpha}{8\beta}}$-progress distribution family $\cD_i$. We now apply Lemma \ref{lemma:pl_prog_distro} along with Theorem \ref{thm:bin_search_general} to conclude the proof of Theorem \ref{cor:bin_search_pl}.
\end{proof}

\subsection{Proof of Theorem \ref{cor:bin_search_dist}}
\label{subsec:dist_proof}

In this subsection, we prove Theorem \ref{cor:bin_search_dist}. % \glcomment{I think this should not be emphasized in the text, since it is a partial result in some sense. reviews also seem to indicate this.}

Again, we present an appropriate interval cover and corresponding sequence of progress distributions $\cD_i(\vx)$ that satisfy the conditions of Theorem \ref{thm:bin_search_general}. See Lemma \ref{lemma:dist_prog_distro}.

\begin{lemma}
\label{lemma:dist_prog_distro}
Fix $i \in \N_{\ge 1}$. Let $\eps = \nfrac{1}{\sqrt{d}} \cdot 2^{-i/2}$. If $\myfunc{f}{\cB_2^d}{\R}$ is $f(\vx) = \norm{\vx-\zstar}_2$ for $\zstar \in \cB_2^d$ and if $\norm{\vx-\zstar}_2 \le \sqrt{2} \cdot \insquare{2^{-(i+1)/2},2^{-i/2}}$, then there exists a distribution $\cD(\vx)$ that can be efficiently sampled from and is a $(1,\gamma, \rho)$-progress distribution for $(\gamma,\rho) = (\nfrac{1}{8},\nfrac{1}{8})$.
\end{lemma}
\begin{proof}[Proof of Lemma \ref{lemma:dist_prog_distro}]
Let $\xplus$ have distribution
\begin{align}
    \frac{\vx - \vg}{\max\inbraces{1, \norm{\vx - \vg}_2}},\quad\quad\text{ where } \vg \sim \eps_t \cdot \mathsf{Unif}(\S_2^{d-1})\label{eqn:projection}.
\end{align}
Note that this distribution can be described as, ``add a uniformly random direction of length $\eps_t$ to $\vx$ and project the result back onto $\actionset = \cB_2^d$.''

It is easy to see that $\norm{\xplus}_2 \le 1$, so the iterates of Algorithm \ref{alg:main_alg} will always remain inside $\cB_2^d$. We now prove that $\cD$ as described above in fact is a $(1,\gamma, \rho)$-progress distribution for the promised parameters.

First, use the fact that $\norm{\vx-\zstar}_2^2$ is $2$-smooth and $2$-P\L{} along with Lemma \ref{lemma:pl_prog_distro} to conclude that
\begin{align*}
    \prvv{\vg}{\frac{\norm{\vx-\zstar}_2^2-\norm{\vx-\vg-\zstar}_2^2}{\norm{\vx-\zstar}_2^2} \ge \frac{2\rho}{d}} \ge \gamma.
\end{align*}
Condition on this event. A basic property of the Euclidean projection onto a convex set implies that
\begin{align*}
    \norm{\xplus-\zstar}_2^2 \le \norm{(\vx-\vg)-\zstar}_2^2
\end{align*}
which yields
\begin{align*}
    \prvv{\vg}{\frac{\norm{\vx-\zstar}_2^2-\norm{\xplus-\zstar}_2^2}{\norm{\vx-\zstar}_2^2} \ge \frac{2\rho}{d}} \ge \gamma.
\end{align*}
Finally, observe that the above event implies
\begin{align*}
    \inparen{\frac{\norm{\xplus-\zstar}_2}{\norm{\vx-\zstar}_2}}^2 \le \inparen{\sqrt{1-\frac{2\rho}{d}}}^2 \le \inparen{1-\frac{\rho}{d}}^2.
\end{align*}
Taking the square root of both sides and rearranging concludes the proof of Lemma \ref{lemma:dist_prog_distro}.
\end{proof}

We remark that the above proof goes through if $\cX$ is an arbitrary convex set; we simply replace \eqref{eqn:projection} with $\Pi_{\cX}(\vx - \vg)$, where $\Pi_{\cX}(\vz)$ is the Euclidean projection of $\vz$ onto $\cX$.

Now, the proof of Theorem \ref{cor:bin_search_dist} will follow in a very similar manner to that of Theorem \ref{cor:bin_search_pl}.

\begin{proof}[Proof of Theorem \ref{cor:bin_search_dist}]
To apply Theorem \ref{thm:bin_search_general}, we need to present $C$, $r$, and a sequence of distribution parameterizations $\cD_i$ that satisfy the premises.

Let $\vx_1 = 0$. It is clear that $\norm{\zstar}_2 \le 1 = B$, which means that the intervals of the form $\sqrt{2} \cdot \insquare{2^{-(i+1)/2},2^{-i/2}}$ for $i \ge 1$ cover the interval $[0,1]$. Hence, for every $i \ge 1$, there exists a corresponding $(1,\nfrac{1}{8},\nfrac{1}{8})$-progress distribution. Theorem \ref{cor:bin_search_dist} follows immediately.
\end{proof}
\section{Proofs of lower bound results} \label{sec:lower-bound}

In this section, we will prove Theorem \ref{thm:large_m_lower} and Corollary \ref{corr:inner-prod-lower-bound}. We first state the following well-known fact (see, e.g., \cite{vershynin2018high}) that there exist $2^{\Omega(d)}$ points inside the unit $\ell_2$ ball which are sufficiently far apart from one another.

\begin{fact}
\label{fact:eps_separated_set}
There exists a subset $S \subset \cB_2^d$ such that $\abs{S} = 2^{\Omega(d)}$, and for all $\vx, \vy \in S$ such that $\vx \neq \vy$, we have $\norm{\vx-\vy}_2 \ge 0.1$.
\end{fact}

We are now ready to prove Theorem \ref{thm:large_m_lower}.

\begin{proof}[Proof of Theorem \ref{thm:large_m_lower}]
We actually prove the lower bound even when the adversary must return the item \emph{in the list} with smallest function value (breaking ties consistently, e.g., according to lexicographic order). Since the adversary is only weaker in this case, this implies the lower bound for the monotone adversary.

By Yao's Lemma~\cite{yao77}, it suffices to give a distribution over instances such that every deterministic algorithm satisfies the conclusions of the theorem. Hence, choose $S$ from Fact \ref{fact:eps_separated_set} and let $\zstar$ be sampled uniformly from $S$.

Fix any deterministic algorithm. The deterministic algorithm branches into at most $m$ states every round, depending on the response the adversary gives. Therefore after $r:=\lfloor\log_m{|S|}\rfloor - 1$ rounds, the algorithm has at most $m^r<\frac12|S|$ distinct states. Each of these states $Q$ can be represented as a tuple of the form $\inbraces{(\xt{1}, \cdots, \xt{m}, i_t) }_{t \in [r]}$, where the $\xt{i}\in \mathcal{X}$ and the $i_t \in [m]$, which represents a set of the algorithm's guesses as well as the closest-point responses for the first $r$ rounds. 

\paragraph{Cost lower bound.} Let us denote $c_r(\zstar, Q)$ to be the total cost incurred for the state $Q$ if the target is $\zstar \in S$. We claim that all but at most one $\zstar \in S$ have $c_r(\zstar,Q)>0.05r$. Suppose there were two points $\zstar$ and ${\zstar}'$ which had $c_r(\zstar, Q) < 0.05r$. Then $c_r(\zstar, Q) + c_r({\zstar}', Q) < 0.1r$, so there exists some round $t \in [r]$ for which
\begin{align*}
\max\inbraces{\norm{\xt{1} - \zstar}, \cdots, \norm{\xt{m} - \zstar}} + \max\inbraces{\norm{\xt{1} - {\zstar}'}, \cdots, \norm{\xt{m} - {\zstar}'}} < 0.1.
\end{align*}
However, this cannot hold by triangle inequality since $\zstar$ and ${\zstar}'$ are well-separated.

For any target $\zstar$, the cost paid in the first $r$ steps is at least $c_r(\zstar,Q(\zstar))$, where $Q(\zstar)$ is the state of the algorithm after $r$ rounds when the target is $\zstar$. In particular, it is of the form $c_r(\zstar, Q)$ for some $Q$. Since there are only $\frac12|S|$ possible algorithm states, at most $\frac12|S|$ values of $\zstar$ can have total cost less than $0.05r$. Therefore, the average cost over instances uniformly drawn from $S$ must be at least 
\[
    \frac{1}{|S|}\left(|S| -\frac12|S|\right)\cdot 0.05r \geq 0.025\left(\frac{\log{S}}{\log{m}}-2\right) = \Omega(\nfrac{d}{\log m}).
\]

\paragraph{Iteration lower bound.} We will use the cost lower bound to prove the iteration lower bound. Recall that we proved that for any algorithm, there existed an instance $\zstar$ for which the algorithm incurs $\Omega(\nfrac{d}{\log m})$ cost over the first $r$ rounds.

Now suppose we had an algorithm $\cA$ which achieved an expected iteration complexity of finding an $\eps$-optimal point of $C \cdot \nfrac{d}{\log m}$ for any $\zstar \in S$, where $\eps, C > 0$ are sufficiently small numerical constants. We can convert this into a low-cost algorithm $\cA'$ for the first $r$ rounds that (1) runs $\cA$ to find an $\eps$-optimal point $\vx$; then (2) until round $r$ repeatedly suggests $\xt{1} = \cdots = \xt{m} = \vx$. The expected cost of algorithm $\cA'$ for the first $r$ rounds is at most
\begin{align*}
    2 \cdot \frac{C d}{\log m} + \eps \cdot r \le \inparen{2C + \eps} \frac{ d}{\log m}.
\end{align*}
For sufficiently small $\eps$ and $C$, we have a contradiction with the previous cost lower bound; thus we can conclude that any algorithm must perform $\Omega\inparen{\nfrac{d}{\log m}}$ iterations in expectation to find an $\eps$-optimal point $\vx$.

This concludes the proof of Theorem \ref{thm:large_m_lower}.
\end{proof}

\begin{proof}[Proof of Corollary \ref{corr:inner-prod-lower-bound}]
The argument for linear $f$ is a reprise of the lower bound for $\ell_2$ distance. Observe that Fact \ref{fact:eps_separated_set} implies that the points in $S$ also satisfy $\ip{\vx, \vy} \le 0.995$ for any $\vx \ne \vy$. 

Therefore, we again use Yao's Lemma and consider deterministic algorithms that branch into $m$ states in every round. Letting $c_r(\zstar, Q)$ denote the total cost incurred for state $Q$ if the target is $\zstar$, we again have the claim that all but at most one $\zstar \in S$ have $c_r(\zstar, Q) > C \cdot r$ for some constant $C > 0$, from which it follows that at most $\frac12|S|$ values of $\zstar$ can have total cost less than $C\cdot r$. We conclude that the average cost over instances drawn uniformly from $S$ must be $\Omega\inparen{\nfrac{d}{\log m}}$.

The argument for the iteration lower bound also proceeds similarly, so we omit the details. 

This concludes the proof of Corollary \ref{corr:inner-prod-lower-bound}.
\end{proof}

\subsection*{Acknowledgements}
AB is supported by NSF Awards CCF-2212968
and ECCS-2216899 and by the Defense Advanced Research Projects Agency under cooperative
agreement HR00112020003. MG is supported by NSF Graduate Research Fellowship. GL is supported by the Institute for Data, Econometrics, Algorithms, and Learning (IDEAL). NSM is supported by NSF Graduate Research Fellowship and NSF Award ECCS-2216899. YY is supported by NSF Award CCF-2045402 and NSF Award CCF-2019844. YY thanks her advisor Jamie Morgenstern for her continued support and encouragement.

\printbibliography

% NSM: just taking a look at some old things
% \input{old/volume_reduction}
% \input{old/local_bin_search}

%\input{noise_model_yyy} % try to prove a lower bound instead. 

\end{document}